\documentclass[a4paper,11pt]{article}

\usepackage{geometry}
\usepackage[latin1]{inputenc}
\usepackage{amsfonts}
\usepackage{amsmath}
\usepackage{amssymb}
\usepackage{amsthm}
\usepackage[T1]{fontenc}
\usepackage{bbm}

\theoremstyle{plain}
\newtheorem{theorem}{Theorem}
\newtheorem{observation}[theorem]{Observation}
\newtheorem{lemma}[theorem]{Lemma}

\newtheorem{proposition}[theorem]{Proposition}

\theoremstyle{definition}
\newtheorem{definition}[theorem]{Definition}
\newtheorem{condition}{Condition}

\theoremstyle{remark}
\newtheorem{remark}{Remark}

\newcommand{\wick}[1]{\text{:}#1\text{:}}

\begin{document}

\title{Local Thermal Equilibrium in Quantum Field Theory on Flat and Curved Spacetimes}
\author{Christoph Solveen\\\quad\\ \small Institut f\"ur Theoretische Physik, Universit\"at G\"ottingen,\\ \small Friedrich-Hund-Platz 1, 37077 G\"ottingen, Germany\\
\quad\\ \small Christoph.Solveen@theorie.physik.uni-goettingen.de
\quad}
\maketitle

\begin{abstract}
The existence of local thermal equilibrium (LTE) states for quantum field theory in the sense of Buchholz, Ojima and Roos is discussed in a model-independent setting. It is shown that for spaces of finitely many independent thermal observables there always exist states which are in LTE in any compact region of Minkowski spacetime. Furthermore, LTE states in curved spacetime are discussed and it is observed that the original definition of LTE on curved backgrounds given by Buchholz and Schlemmer needs to be modified. Under an assumption related to certain unboundedness properties of the pointlike thermal observables, existence of states which are in LTE at a given point in curved spacetime is established. The assumption is discussed for the sets of thermal observables for the free scalar field considered by Schlemmer and Verch.
\end{abstract}

\renewcommand\abstractname{Mathematics Subject Classifaction (2010)}

\begin{abstract}
  81T20, 81T28, 82C99
\end{abstract}

\section{Introduction}

In quantum field theory (QFT) it is well known that states which represent ensembles in global thermal equilibrium are to be described with the help of the KMS condition \cite{Haag96}. These states have many nice properties, among them the fact that they may be characterized uniquely with the help of only a few "thermal parameters", for example temperature. Buchholz, Ojima and Roos \cite{BuchOjimRoos02} presented a method to identify non-equilibrium states to which one may still attach at least some of these thermal parameters locally. This is done with the help of so-called thermal observables, the expectation values of which are taken as an indicator of when a state is locally in thermal equilibrium.

While there are interesting examples of local thermal equilibrium (LTE) states for free fields (see \cite{Buch03}, for example), it is useful to study existence of LTE states without referring to any particular model. A first investigation of this matter has been carried out in \cite{BuchOjimRoos02}, where it was shown that for certain spaces of thermal observables there are states which are in LTE at a given point of spacetime. In section \ref{Mink}, we improve on this argument and show that for any compact region $\mathcal{O}$ of Minkowski spacetime there are states which are in LTE in $\mathcal{O}$\,. Thermal parameters attached to these states exhibit non-trivial spacetime dependence.

Making use of locally covariant QFT, Buchholz and Schlemmer \cite{BuchSchl07} described a way to identify LTE states in QFT on curved spacetimes. In section \ref{Curved}, we investigate the structure of the sets of thermal observables in curved spacetime, finding that the sets proposed in \cite{BuchSchl07} are in general too restrictive to allow LTE, which is illustrated in a model considered by Schlemmer and Verch in \cite{SchlVerc08}. We give a modified definition of LTE on curved spacetime and in section \ref{Curved2} we proceed to show that under a natural assumption on the sets of thermal observables there are states which are in LTE at a given point of spacetime. Here, we employ methods completely unrelated to the ones we use in finding LTE states in Minkowski spacetime. In an appendix we sketch how our assumption can be verified in the model considered previously.

We note that in \cite{HollLeil10} Hollands and Leiler discuss the Boltzmann transport equation within QFT. It would be interesting to investigate the relationship of their work with our approach to LTE.

\section{Local Equilibrium on Minkowski Spacetime}\label{Mink}

We begin by recalling the basic notions of LTE for QFT on Minkowski spacetime \cite{BuchOjimRoos02}. The underlying idea is that a state which describes an ensemble in LTE at a given point $x$ should not be distinguishable from some (mixture of) global thermal equilibrium states with regard to measurements performed with ``thermal observables'' localized at $x$. The latter are taken from suitable subsets of the field content of the QFT, whereas the global thermal equilibrium states are identified by use of the KMS condition.

We briefly describe the set-up in order to fix our notation, details can be found in \cite{BuchOjimRoos02}. We assume that the $*$-algebra of observables, $\mathcal{A}$\,, is generated by the smearings of a countable set $\mathbf{F}$ of observable (Hermitean) quantum fields:
\begin{equation}
\mathbf{F}=\bigl\{\phi_{i}\bigr\}_{i\in I}~~~~~\text{($I$ some countable index set)}\,.
\end{equation}
We call $\mathbf{F}$ the \emph{field content}. The fields $\phi_{i}$ can be of arbitrary tensorial nature, which in this section will not be indicated in our notation.  Not all of the quantum fields $\phi_{i}$ have to be ``fundamental''; they may be ``derived`` in the sense that they are (differentiated) normal products of one or more basic fields (think, for example, of the algebra of Wick polynomials of the free scalar field \cite{HollWald01}). The Poincar\'e group acts on $\mathcal{A}$ by automorphisms, which in case of pure translations by $x\in\mathbb{R}^{4}$ we denote by $\alpha_{x}$\,.

We assume that a physically reasonable subset $\mathcal{S}$ of the set of states of $\mathcal{A}$ has been chosen. Since we are interested in measuring pointlike fields, $\mathcal{S}$ has to be such that for all $\omega\in\mathcal{S}$ and all $i\in I$ the distributions $f\mapsto\omega(\phi_{i}(f))$ can be represented by suitable (at least continuous) functions $\omega(\phi_{i}(x))$\,. The point-like fields $\phi_{i}(x)$ can then be understood as linear forms on the linear span of $\mathcal{S}$\,.

The KMS condition singles out states in $\mathcal{S}$ which describe ensembles in global thermal equilibrium in their rest system, the latter being described by a timelike unit vector $e_{0}$\,, at some inverse temperature $\beta$\,. We encode the information about the rest frame and the temperature into a single timelike vector $\beta\cdot e_{0}$, which we denote by $\beta$ once again for brevity. We make two important assumptions on the set of KMS states. Firstly, for each $\beta\in V^{+}$ ($V^{+}$ being the forward light cone), we assume that there is a unique KMS state $\omega_{\beta}$ \footnote{This means that the systems we consider consist of a single phases only.}, implying that these states are isotropic and invariant under translations. Moreover, we assume that the KMS states $\omega_{\beta}$ are weakly continuous in $\beta$, i.e.\ we require that the functions
\begin{equation*}
\beta\mapsto\omega_{\beta}(a)
\end{equation*}
are continuous for all $a\in\mathcal{A}$\,. Thus, mixtures of KMS states may be formed by means of integration with suitable positive normalized measures $\rho$ on $V^{+}$\,:
\begin{equation}\label{ReferenceStates}
 \omega_{\rho}:=\int\,d\rho(\beta)\,\omega_{\beta}\,.
\end{equation}
The set of \emph{thermal reference states}, denoted by $\mathcal{C}$\,, is defined as the collection of all states of this form.

Since by assumption for each $\beta\in V^{+}$ there exists a unique KMS state $\omega_{\beta}$\,, every intensive thermodynamic quantity $F$ attached to these states (thermal energy density, entropy density etc.) can be expressed as a function of $\beta$ alone. These functions $\beta\mapsto F(\beta)$ are called \emph{thermal functions}. They characterize the macroscopic properties of the KMS states $\omega_{\beta}$ completely. Thus we may write $F(\beta)=\omega_{\beta}(F)$ and interpret $F$ as a macroscopic (central) observable. Evaluated in a reference state $\omega_{\rho}\in\mathcal{C}$\,, one obtains the mean value of $F$ with regard to the measure $\rho$\,:
\begin{displaymath}
 \omega_{\rho}(F)=\int\,d\rho(\beta)\,F(\beta)\,.
\end{displaymath}
Hence, thermodynamic quantities may also be attached to the reference states.

As is explained in \cite{BuchOjimRoos02}, not all members of the field content $\mathbf{F}$ are sensitive to thermal properties of the reference states, so one needs to choose a suitable subset of \emph{thermal observables} $\mathbf{T}\subset\mathbf{F}$\,, which we take to include the unit $\mathbbm{1}$\,. The real vector space that is generated by $\mathbf{T}(x)$ will be denoted by $S(x)$\,. Here, the set $\mathbf{T}(x)$ consists of the elements of $\mathbf{T}$\,, evaluated at $x$ in some reference frame.

In order to establish a link between microscopic and macroscopic properties of the reference states, one considers the particular subset of thermal functions which are obtained by evaluating the elements of $\mathbf{T}(x)$ in the KMS states:
\begin{equation}\label{ThermalFunctions}
 \Phi(\beta):=\omega_{\beta}\bigl(\phi(x)\bigr)\,.
\end{equation}
Ideally, $S(x)$ is large enough so that one is able to reproduce (or at least approximate) relevant thermal functions $F(\beta)$ with some $\Phi(\beta)$ as in \eqref{ThermalFunctions}. It follows from our assumptions that the functions $\Phi$ are continuous in $\beta$ \footnote{Take any test function $f$ with $\int f(x)\,\mathrm{d}^4x=1$\,. By assumption, $\omega_{\beta}(\phi(f))$ is continuous in $\beta$\, and by translation invariance of the KMS states, one has $\omega_{\beta}(\phi(f))=\omega_{\beta}(\phi(x))=\Phi(\beta)$\,, hence $\beta\mapsto\Phi(\beta)$ is continuous.}.
\begin{definition}
 A state $\omega\in\mathcal{S}$ is called \emph{$S(x)$-thermal} if there exists a normalized measure $\rho_{x}$\,, supported in $V^{+}$\,, such that (cf.\ \eqref{ReferenceStates})
\begin{equation}\label{LTEstates}
 \omega(\phi(x))=\omega_{\rho_{x}}(\phi(x))
\end{equation}
for all $\phi(x)\in S(x)$.
\end{definition}
Note that the $x$-dependence of the right-hand side of \eqref{LTEstates} lies entirely in the measure $\rho_{x}$, since the reference states in $\mathcal{C}$ are translationally invariant by assumption.

Given an $S(x)$-thermal state $\omega$, one may determine the mean values of the corresponding thermal functions (cf.\ \eqref{ThermalFunctions}) in $\omega$ at the spacetime point $x$\,:
\begin{equation}\label{ThermalFunctionsLTE}
 \omega(\Phi)(x):=\omega(\phi(x))~.
\end{equation}
Due to $S(x)$-thermality, this provides a consistent lift of $\omega$ to the space of thermal functions and thereby opens the possibility to assign thermodynamic quantities to LTE states.

In order to consider states in which the thermal functions can vary over spacetime, one may extend the definition to encompass states that are in LTE in a whole region $\mathcal{O}$ of Minkowski spacetime. For each $x\in\mathcal{O}$\,, let the space $S(x)$ be as described before. A state $\omega$ is called \emph{$S_{\mathcal{O}}$-thermal} if $\omega$ is $S(x)$-thermal for all $x\in\mathcal{O}$\,. The spatio-temporal behaviour of the thermal functions in the state $\omega$ on $\mathcal{O}$ is encoded in the functions $x\mapsto\omega(\Phi)(x)$ obtained via \eqref{ThermalFunctionsLTE}. In \cite{BuchOjimRoos02} it is shown how the microscopic dynamics of the thermal observables may thus give rise to equations governing the spacetime dependence of the macroscopic observables.

\quad

In the following we wish to discuss whether there actually exist any non-trivial LTE states, apart from the obvious ones in $\mathcal{C}$\,. Let $S(0)$ be the space generated by the thermal observables $\mathbf{T}(0)$ at the origin $0$ of Minkowski spacetime. We assume that $\mathbf{T}(0)$ is finite, so that $S(0)$ is finite-dimensional. Then, for each $x$ in some region $\mathcal{O}$\,, $S(x)=\alpha_{x}S(0)$ has the same dimension. For each $x\in\mathcal{O}$ and any compact $B\subset V^{+}$ we introduce the following seminorm on $S(x)$ \footnote{We call each of these seminorms $\tau_{B}$, disregarding on which space $S(x)$ they are defined.}:
\begin{equation}\label{SemiNormTau}
 \tau_{B}\bigl(\phi(x)\bigr):=\sup_{\beta\in B}|\omega_{\beta}\bigl(\phi(x)\bigr)|\,.
\end{equation}
We make use of these maps in:

\begin{proposition}
Let $\mathcal{O}$ be a compact subregion of Minkowski spacetime. Assume that there is some compact $B\subset V^{+}$ such that $\tau_{B}$ constitutes a norm on $S(0)$\,. Then, there exist $S_{\mathcal{O}}$-thermal states in $\mathcal{S}$ which are not in $\mathcal{C}$ (except when $\mathcal{S}=\mathcal{C}$ already).
\end{proposition}

\begin{proof}
We pick any state $\omega_{0}\in\mathcal{S}$ with $\omega_{0}\notin\mathcal{C}$\,. In the following, we use $\omega_{0}$ to construct an LTE state $\omega$ that is not in $\mathcal{C}$\,.

Since $\tau_{B}$ is a norm on $S(0)$\,, it is also on each $S(x)$\,, $x\in\mathcal{O}$\,. Both this and the finite-dimensionality of $S(x)$ imply that for each $x\in\mathcal{O}$ there is some positive real number $C_{x}$ such that
\begin{equation}
 |\omega_{0}(\phi(x))|\leq C_{x}\,\tau_{B}(\phi(x))~~\text{for all}~\phi(x)\in S(x)\,.
\end{equation}

This in turn implies that for each $x\in\mathcal{O}$ we can lift $\omega_{0}$ to become a linear functional on the subspace of the thermal functions corresponding to $S(x)$ by \eqref{ThermalFunctions}. This is done by setting $\omega_{0}(\Phi)(x):=\omega_{0}(\phi(x))$\,, where $\Phi$ is the thermal function corresponding to $\phi(x)$ (cf.\ \eqref{ThermalFunctions}). In particular $\Phi\equiv 0$ entails $\omega_{0}(\Phi)(x)=0$\,, so the (normalized) functionals $\omega_{0}(\cdot)(x)$ are well-defined for each $x\in\mathcal{O}$\,. Moreover, they are bounded by the norm $\tau_{B}$ and thus we may extend them to the whole of $C^{0}(B)$ by use of the Hahn-Banach theorem \footnote{Here, $C^{0}(B)$ denotes the space of continuous complex-valued functions on $B$.}. We may choose these extensions, denoted by $\omega_{0}(\cdot)(x)$ once again, to be Hermitean. For each $x\in\mathcal{O}$ we find
\begin{equation*}
 |\omega_{0}(F)(x)|\leq C_{x}\,||F||_{B}~~\text{for all}~F\in C^{0}(B)\,,
\end{equation*}
where $||\cdot||_{B}$ is the supremum norm on $C^{0}(B)$\,, with respect to which this space is a commutative $C^{*}$-algebra with unit. Any normalized Hermitean linear functional on this algebra can be represented by a normalized \emph{signed} measure $\sigma_{x}$ \cite{Cohn80}, decomposable into two positive measures: $\sigma_{x}=\sigma^{+}_{x}-\sigma^{-}_{x}$\,. Thus, for each $x\in\mathcal{O}$\,,
\begin{equation}\label{MeasureRep}
 \omega_{0}(F)(x)=\int_{B}~d\sigma_{x}(\beta)~F(\beta)\,.
\end{equation}

Assume that we are given a positive measure $\tau$ such that the sum $\sigma_{x}+\tau$ is a \emph{positive} measure for all $x\in\mathcal{O}$ and let us define the state
\begin{equation}\label{LTEStateO}
 \omega:=(1+\tau(B))^{-1}\Bigl(\omega_{0}+\int_{B}\,d\tau(\beta)\,\omega_{\beta}\Bigr)~.
\end{equation}
It is easy to see that for any $x\in\mathcal{O}$ and for all $\phi(x)\in S(x)$\,, we have $\omega(\phi(x))=\int_{B}\,d\rho_{x}(\beta)\,\omega_{\beta}(\phi(x))$ with measure 
\begin{equation}\label{MeasureTau}
\rho_{x}:=(1+\tau(B))^{-1}\bigl(\sigma_{x}+\tau\bigr)\,.
\end{equation}
In other words, $\omega$ is an $S(x)$-thermal state for all $x\in\mathcal{O}$, i.e.\ $\omega$ is $S_{\mathcal{O}}$-thermal. It is not a reference state in general, since $\omega_{0}$ is not.

It remains to construct a suitable measure $\tau$\,. We have obtained the functionals $\omega_{0}(\cdot)(x)$ for each $x$ on the space of thermal functions corresponding to elements of $S(0)$\,. This space is a finite dimensional vector space and hence we may express each of the $\omega_{0}(\cdot)(x)$ in terms of a basis of Hermitean functionals $\{\omega^{i}\}^{n}_{i=1}$\,:
\begin{displaymath}
 \omega_{0}(\cdot)(x)=\sum^{n}_{i=1}a_{i}(x)\,\omega^{i}
\end{displaymath}
for some $a_{i}(x)\in\mathbb{R}$\,. Due to the continuity of $x\mapsto\omega_{0}(\phi(x))$ for all $\phi(x)\in S(x)$ (which is guaranteed by our choice of $\mathcal{S}$), we find that $x\mapsto\omega_{0}(\Phi)(x)$ is continuous for all thermal fucntions $\Phi$ induced by $S(x)$\,. By compactness of $\mathcal{O}$, the coefficient functions $a_{i}:\mathcal{O}\rightarrow\mathbb{R}$ are therefore bounded by positive constants $C_{i}$\,. We may use the Hahn-Banach theorem to extend the functionals $\omega^{i}$ to Hermitean functionals on all of $C^{0}(B)$\,, which we also denote by $\omega^{i}$. As before, each of the $\omega^{i}$ corresponds to a signed measure $\sigma^{i}=\sigma^{i,+}-\sigma^{i,-}$ and it follows that the measures $\sigma_{x}$ in \eqref{MeasureRep} can be written as $\sigma_{x}=\sum^{n}_{i=1}a_{i}(x)\,\sigma^{i}$\,. Therefore, the $x$-dependence of the signed measures $\sigma_{x}$ is given entirely by the functions $a_{i}$\,. We now define the positive measure
\begin{displaymath}
 \tau:=\sum_{i=1}^{n}C_{i}\,|\sigma^{i}|\,,
\end{displaymath}
where $|\sigma^{i}|:=\sigma^{i,+}+\sigma^{i,-}$ denote the \emph{variations} of the signed measures $\sigma^{i}$ (which are positive by definition). It is easy to see that $\sigma_{x}+\tau$ is a positive measure for all $x\in\mathcal{O}$\,.
\end{proof}

The LTE states we have constructed are of the form indicated in \eqref{LTEStateO} and may therefore be interpreted as a ``perturbation'' (given by $\omega_{0}$) around the thermal reference state $\frac{1}{\tau(B)}\int_{B}\,d\tau(\beta)\,\omega_{\beta}$ with appropriate measure $\tau$\,. In these states, the thermal functions corresponding to $S(0)$ will exhibit non-trivial spacetime dependence as may be seen by the fact that $\omega(\Phi)(x)=\omega_{\rho_{x}}(\phi(0))$ for all $x\in\mathcal{O}$\,. The measure $\rho_{x}$, given by equation \eqref{MeasureTau}, clearly depends non-trivially on $x$ because the function $x\mapsto\omega_{0}(\phi(x))$ depends non-trivially on $x$\,.

It is shown in \cite{BuchOjimRoos02} that the argument for the existence of an $S(x)$-thermal state for one \emph{point} $x$ can be extended to the case where $\tau_{B}$ is not a norm. This case corresponds to the occurrence of ``equations of state'', i.e.\ non-trivial relations between the thermal functions. It is unclear whether the same reasoning may be adopted in case of $S_{\mathcal{O}}$-thermality. It appears, however, that the local validity of some equation of state is a strong restriction for LTE states, and therefore the assumption that $\tau_{B}$ is a norm seems appropriate.

\section{Local Thermal Observables in Curved Spacetime}\label{Curved}

Buchholz and Schlemmer \cite{BuchSchl07} made a proposal on how to define LTE for quantum field theories on curved spacetime. The definition of a sensible class of reference states requires global thermal equilibrium states, which in general do not exist in the presence of curvature. The authors therefore suggested using \emph{locally covariant quantum field theory} \cite{BrunFredVerc03} to be able to compare states of QFTs on different spacetimes. This enables one to use the reference states on Minkowski spacetime in order to define LTE on curved spacetime. We find that the original definition of LTE has to be modified, because in the presence of curvature, linear combinations of thermal observables cannot in general be thermal observables again. In this section, we discuss this important point in some detail while in the next section we show existence of states that are in LTE at any given point of curved spacetime.

We write $(M,g)$ for a 4-dimensional (globally hyperbolic, oriented, time-oriented) spacetime $M$ with Lorentzian metric $g$. We assume that we are given a countable set $\mathbf{F}$ of ``master fields'' that can be evaluated on each $(M,g)$ to yield a set
\begin{equation}\label{FieldContent}
 \mathbf{F}_{g}=\bigl\{\phi^{(i)}_{g}\bigr\}_{i\in I}~~~~~~~\text{($I$ some countable index set)}
\end{equation}
of observable, tensorial quantum fields \footnote{Note that here and in the following, objects carrying a subscript ``$g$'' refer to the spacetime $(M,g)$\,.} on $(M,g)$\,. We take $\phi^{(0)}_{g}$ to be the identity $\mathbbm{1}$ and call $\mathbf{F}_{g}$ the field content on $(M,g)$\,. The smearings of the elements of $\mathbf{F}_{g}$ generate a $*$-algebra $\mathcal{A}(M,g)$\,, the algebra of observables on $(M,g)$\,.

For definiteness, let the quantum fields correspond to contravariant tensors of rank $r(i)$\,. If $A_{g,(i)}$ is a smooth covariant tensor field of rank $r(i)$ on $(M,g)$\,, then the evaluation $A_{g,(i)}\cdot\phi^{(i)}_{g}$ yields a (scalar) distribution taking values in $\mathcal{A}(M,g)$ (there is no sum over the indices $i$\,). Naturally, if $e$ is a frame represented by some Lorentz tetrad \footnote{This means that the tetrad is orthonormal with respect to $\mathrm{diag}(+1,-1,-1,-1)$ and $e_{0}$ is timelike future pointing.} $\{e_{\mu}\}^{3}_{\mu=0}$\,, then
\begin{equation*}
 \phi^{(i)}_{g,\,\mu_{1}\dots\mu_{r(i)}}:= e_{\mu_{1}}\otimes\dots\otimes e_{\mu_{r(i)}}\cdot\phi^{(i)}_{g}
\end{equation*}
are the \emph{components} of $\phi^{(i)}_{g}$ in the frame $e$\,.

For each $(M,g)$, we select a convex subset $\mathcal{S}(M,g)$ of the set of states of $\mathcal{A}(M,g)$\,, which we assume to be such that for each $\omega\in\mathcal{S}(M,g)$ the distribution $$f\mapsto\omega(A_{g,(i)}\cdot\phi^{(i)}_{g}(f))$$ can in fact be represented by a (smooth) function $x\mapsto\omega(A_{g,(i)}\cdot\phi^{(i)}_{g}(x))$. In other words, for each $x\in M$ we may view the point-like fields $A_{g,(i)}\cdot\phi^{(i)}_{g}(x)$ as linear forms on the linear span of $\mathcal{S}(M,g)$\,.

Roughly, the principles of locality and covariance \cite{BrunFredVerc03} require (among other things) that for each orientation, time orientation and causality preserving isometric embedding \footnote{We call such maps \emph{causal isometric embeddings}.} $\psi:(M,g)\rightarrow (M',g')$\,, there exists an injective $*$-homomorphism $\alpha_{\psi}:\mathcal{A}(M,g)\rightarrow\mathcal{A}(M',g')$ that extends to the point-like fields such that
\begin{equation}\label{LCQF}
 \alpha_{\psi}(A_{g,(i)}\cdot\phi^{(i)}_{g}(x))=A_{g',(i)}\cdot\phi^{(i)}_{g'}(\psi(x))~,
\end{equation}
where $A_{g',(i)}:=\psi_{*}A_{g,(i)}$ is the natural push-forward of $\psi$ acting on the tensor field $A_{g,(i)}$\,. The assignment $(M,g)\mapsto\phi^{(i)}_{g}$ is called a \emph{(tensorial) local and covariant quantum field}. In addition we require $\alpha_{\psi}^{*}\mathcal{S}(M',g')\subset\mathcal{S}(M,g)$\,, where $\alpha_{\psi}^{*}$ is the dual map of $\alpha_{\psi}$\,, which lets the map $(M,g)\mapsto\mathcal{S}(M,g)$ become a \emph{locally covariant state space}.

\quad

After these preparations, let us come to our definition of LTE in curved spacetime, which, as was mentioned before, is different from the one given in \cite{BuchSchl07}. On Minkowski spacetime, denoted by $(M_{0},g_{0})$\,, we choose the set of thermal reference states $\mathcal{C}\subset\mathcal{S}(M_{0},g_{0})$ as in section \ref{Mink}, satisfying the same assumptions. We select a \emph{finite} set of thermal observables
\begin{equation}\label{TO}
  \bigl\{\tau^{(i)}_{g_{0}}\bigr\}^{n}_{i=0}\,,~\tau^{(0)}_{g_{0}}=\mathbbm{1}\,,
\end{equation}
as a subset of the field content on Minkowski spacetime. Now we choose a subset $\mathbf{T}$ of the set of ``master fields'' $\mathbf{F}$ such that $\mathbf{T}_{g_{0}}$ is exactly the set of thermal observables \eqref{TO}. On each $(M,g)$\,, we therefore find a subset $\mathbf{T}_{g}$ of the field content, which may be used as the set of thermal observables.

\begin{remark}
In general, there are many choices for $\mathbf{T}$\,. Hence there is some leeway in determining thermal observables on curved spacetime and criteria are needed to make a suitable choice. In the particular example of the algebra of Wick polynomials of the free scalar field \cite{HollWald01}, the (differentiated) Wick powers are determined up to some universal renormalization constants. An interesting method to fix these numbers is presented in \cite{BuchSchl07}, using as additional physical input the interplay between the concept of LTE and KMS states on particular spacetimes where the latter exist.
\end{remark}

For $x\in M$\,, we define the set $\mathbf{T}_{g}(x)$ to consist of all the components of the fields $\tau^{(i)}_{g}$\,, $i=0,1,\dots,n$\,, in a local Lorentz frame $e_{g}$ around $x$. At $x$\,, we identify $e_{g}$ with some Lorentz frame $e_{g_{0}}$ at the origin $0$ of Minkowski spacetime and define $\mathbf{T}_{g_{0}}(0)$ accordingly with respect to $e_{g_{0}}$\,.

\begin{definition}\label{LTEcurved}
A state $\omega\in\mathcal{S}(M,g)$ is \emph{$\mathbf{T}_{g}(x)$-thermal} if for some $\omega_{\rho_{x}}\in\mathcal{C}$
\begin{equation*}
 \omega(\tau^{(i)}_{g,\,\mu_{1}\dots\mu_{r(i)}}(x))=\omega_{\rho_{x}}(\tau^{(i)}_{g_{0},\,\mu_{1}\dots\mu_{r(i)}}(0))
\end{equation*}
for all $\mu_{1},\dots,\mu_{r(i)}$ for all $i=0,1,\dots,n$.
\end{definition}

\begin{remark}
 The identification of the two frames $e_{g}$ and $e_{g_{0}}$ is arbitrary. The definition of LTE is however not affected by this choice, as seen in the following example. Let $\omega\in\mathcal{S}(M,g)$ be in LTE at $x$ with regard to some thermal observable $\tau_{g,\,\mu}(x)$ under the identification of $e_{g}$ and $e_{g_{0}}$\,. For simplicity, suppose $\omega$ has a definite temperature vector $\beta\in V^{+}$ at $x$. This means
\begin{equation}\label{ExLTE}
 \omega(\tau_{g,\,\mu}(x))=\omega_{\beta}(\tau_{g_{0},\,\mu}(0))~\text{for some}~\omega_{\beta}\in\mathcal{C}~.
\end{equation}
Choosing another frame $e'_{g}=e_{g}\Lambda^{-1}$ ($\Lambda$ some (proper) local Lorentz transformation in $T_{x}M$), we find
\begin{align*}
 \omega(e'_{g,\,\mu}\cdot\tau_{g}(x))&\,\,=\,\,(\Lambda^{-1})_{~\,\mu}^{\nu}\,\omega(e_{g,\,\nu}\cdot\tau_{g}(x))\\
&\stackrel{\eqref{ExLTE}}{=}(\Lambda^{-1})_{~\,\mu}^{\nu}\,\omega_{\beta}(e_{g_{0},\,\nu}\cdot\tau_{g_{0}}(0))\\
&\,\,=\,\,\omega_{\Lambda\beta}(e_{g_{0},\,\mu}\cdot\tau_{g_{0}}(0))~,
\end{align*}
where the last equality follows from the transformation law of the KMS states under the Poincar\'e group \cite[eqn.\ (2.8)]{BuchOjimRoos02}. Had we now identified $e'_{g}$ with $e_{g_{0}}$\,, we would still conclude that $\omega$ is in LTE with the same scalar temperature $\beta\in\mathbb{R}^{+}$\,.
\end{remark}

We would like to emphasize the following important point: contrary to the definition of LTE on curved spacetimes given in \cite{BuchSchl07}, where whole vector spaces of thermal observables are considered, we find that \emph{linear combinations of thermal observables cannot in general be thermal observables again}. We will exemplify this in the following model, which is treated in detail in \cite{SchlVerc08}.

Namely, on some fixed $(M,g)$\,, we consider the Klein Gordon field (massless, for simplicity): $(\Box_{g}+\xi R_{g})\phi=0$ \footnote{Here, $\Box_{g}$ denotes the D'Alembertian on $(M,g)$ while $R_{g}$ is the scalar curvature.} with curvature coupling $\xi\in\mathbb{R}$\,. The field content $\mathbf{F}_{g}(x)$ at $x$  in some spacetime $(M,g)$ consists of the local and covariant (differentiated) Wick powers
\begin{equation*}
 \wick{(\nabla^{o(1)}\phi)(\nabla^{o(2)}\phi)\dots (\nabla^{o(n)}\phi)}
\end{equation*}
and their total derivatives in some frame $e_{g}$ (as constructed in \cite{HollWald01}, see \cite{More03} for differentiated Wick powers) \footnote{We assume that renormalization ambiguities are dealt with already.}. The $\nabla^{o}$ are monomials in the covariant derivatives of order $o$\,. The appropriate state space $\mathcal{S}(M,g)$ is given by the set of Hadamard states \cite{Wald94, BrunFredVerc03, SahlVerc01, HollRuan02, Sand09}, in which the pointlike fields can be defined.

Fix $x\in M$\,. The set of thermal observables at $x$ is defined as
\begin{equation}\label{SetTO}
 \mathbf{T}^{(2)}_{g}(x)=\{\mathbbm{1},\wick{\phi^2}_{g}(x),\{\epsilon_{g,\mu\nu}(x)\}\}\,,
\end{equation}
where
\begin{equation}
 \epsilon_{g,\mu\nu}(x):=-\wick{\phi(\nabla_{\mu}\nabla_{\nu}\phi)}_{g}(x)+\frac{1}{4}\nabla_{\mu}\nabla_{\nu}\wick{\phi^{2}}_{g}(x)
\end{equation}
are the components of the \emph{thermal energy tensor}, see \cite{BuchOjimRoos02, SchlVerc08} (using the Leibniz rule, one can express this tensor in different ways).

After choosing an identification of frames as described previously, we have the following assignments of forms due to local covariance:
\begin{equation}\label{Assign1}
\begin{split}
&\mathbbm{1}\leftrightarrow\mathbbm{1}\,,~\wick{\phi^2}_{g}(x)\leftrightarrow\wick{\phi^2}_{g_{0}}(0)~\text{and}\\
&\epsilon_{g,\mu\nu}(x)\leftrightarrow\epsilon_{g_{0},\mu\nu}(0)~\text{for}~\mu,\nu=0,1,2,3\,.
\end{split}
\end{equation}
This provides an identification of $\mathbf{T}^{(2)}_{g_{0}}(0)$ with $\mathbf{T}^{(2)}_{g}(x)$\,, which we need in order to define LTE on $M$\,. There is, however, another locally covariant quantum field of interest, since it arises as linear combination of thermal observables \footnote{Recall that $\eta^{\mu\nu}\nabla_{\mu}\nabla_{\nu}\wick{\phi^2}_{g}(x)=g^{ab}\nabla_{a}\nabla_{b}\wick{\phi^2}_{g}(x)$ in abstract index notation, since the frame is taken to be orthonormal at $x$\,.}:
\begin{equation}\label{Assign2}
 \eta^{\mu\nu}\nabla_{\mu}\nabla_{\nu}\wick{\phi^2}_{g}(x)\leftrightarrow\eta^{\mu\nu}\partial_{\mu}\partial_{\nu}\wick{\phi^2}_{g_{0}}(0)\,.
\end{equation}
The corresponding relations read \cite{SchlVerc08}:
\begin{equation}\label{Relations}
 \begin{split}
\eta^{\mu\nu}\epsilon_{g,\mu\nu}(x)-U_{g}(x)\mathbbm{1}-\xi R_{g}(x)\,\wick{\phi^2}_{g}(x)&=\frac{1}{4}\eta^{\mu\nu}\nabla_{\mu}\nabla_{\nu}\wick{\phi^2}_{g}(x)\\
\text{and}\quad\eta^{\mu\nu}\epsilon_{g_{0},\mu\nu}(0)&=\frac{1}{4}\eta^{\mu\nu}\partial_{\mu}\partial_{\nu}\wick{\phi^2}_{g_{0}}(0)\,,
\end{split}
\end{equation}
 $\eta^{\mu\nu}$ being the entries of the diagonal matrix $\mathrm {diag}(1,-1,-1,-1)$\,. Here, $U_{g}(x)$ is determined by $\xi$ and the local geometry of $(M,g)$ around $x$ \cite{SchlVerc08,More03}.

It is a consequence of these relations that the assignments \eqref{Assign1} and \eqref{Assign2} cannot be extended to a \emph{linear} map between $\mathrm{span}_{\mathbb{R}}\{\mathbf{T}^{(2)}_{g_{0}}(0)\}$ and $\mathrm{span}_{\mathbb{R}}\{\mathbf{T}^{(2)}_{g}(x)\}$\,, because this would require the term $U_{g}(x)\mathbbm{1}+\xi R_{g}(x)\,\wick{\phi^2}_{g}(x)$ to vanish, which it clearly does not. Therefore, the locally covariant quantum fields only induce an identification between the sets $\mathbf{T}^{(2)}_{g_{0}}(0)$ and $\mathbf{T}^{(2)}_{g}(x)$\,, and \emph{not} between the (real) vector spaces spanned by them.

This situation can be expected to occur quite generally, because there are linear relations, arising for dynamical reasons like \eqref{Relations}, between the fields in any physical model. The point is that these relations, in general, look different at $x\in M$ and $0\in M_{0}$\,, which prevents us from identifying \emph{linear} spaces of thermal observables at $x$ and $0$ using locally covariant quantum fields. Therefore the definition of LTE on curved spacetimes given in \cite{BuchSchl07} has to be relaxed, because it requires a linear correspondence between vector spaces of thermal observables at $x$ and $0$ \footnote{There are non-linear (non-unique) correspondences, but this collides with the definition of LTE states, because states are linear.}.

However, let $\omega\in\mathcal{S}(M,g)$ be $\mathbf{T}^{(2)}_{g}(x)$-thermal (cf.\ Definition \ref{LTEcurved}). Then, by \eqref{Relations}, using the fact that the reference states are translationally invariant
\begin{equation}\label{RelLTE}
 \begin{split}
0&=\frac{1}{4}\omega_{\rho}\bigl(\eta^{\mu\nu}\partial_{\mu}\partial_{\nu}\wick{\phi^2}_{g_{0}}(0)\bigr)=\eta^{\mu\nu}\omega_{\rho}\bigl(\epsilon_{g_{0},\mu\nu}(0)\bigr)\\
&=\eta^{\mu\nu}\omega\bigl(\epsilon_{g,\mu\nu}(x)\bigr)=U_{g}(x)+\xi R_{g}(x)\,\omega\bigl(\wick{\phi^2}_{g}(x)\bigr)+\frac{1}{4}\omega\bigl(\eta^{\mu\nu}\nabla_{\mu}\nabla_{\nu}\wick{\phi^2}_{g}(x)\bigr)\\
&=U_{g}(x)+\xi R_{g}(x)\,\omega_{\rho}\bigl(\wick{\phi^2}_{g_{0}}(0)\bigr)+\frac{1}{4}\omega\bigl(\eta^{\mu\nu}\nabla_{\mu}\nabla_{\nu}\wick{\phi^2}_{g}(x)\bigr)\,.
\end{split}
\end{equation}
The important thing to note here is that $\eta^{\mu\nu}\nabla_{\mu}\nabla_{\nu}\wick{\phi^2}_{g}(x)$ is not a thermal observable (even though it is a sum of thermal observables). The condition of $\mathbf{T}^{(2)}_{g}(x)$-thermality does not require it to take the same value in $\omega$ as $\eta^{\mu\nu}\partial_{\mu}\partial_{\nu}\wick{\phi^2}_{g_{0}}(0)$ does in the reference state $\omega_{\rho}$ (which is zero), but it fixes its value by \eqref{RelLTE}.

\begin{remark}
 We have not defined LTE in regions of curved spacetime (as we did in the case of Minkowski spacetime). Assume for a moment that there is a region $\mathcal{O}$ of $M$ such that $\omega$ is $\mathbf{T}^{(2)}_{g}(x)$-thermal for all $x\in\mathcal{O}$\,. Using the fact that
\begin{equation}
\begin{split}
\omega\bigl(\eta^{\mu\nu}\nabla_{\mu}\nabla_{\nu}\wick{\phi^2}_{g}(x)\bigr)&=\Box_{g}\omega\bigl(\wick{\phi^2}_{g}(x)\bigr)\\
&=\Box_{g}\omega_{\rho_{x}}\bigl(\wick{\phi^2}_{g_{0}}(0)\bigr)
\end{split}
\end{equation}
by the LTE condition, we find that by \eqref{RelLTE}
\begin{equation}
 \left(\frac{1}{4}\Box_{g}+\xi R_{g}(x)\right)\,\omega_{\rho_{x}}\bigl(\wick{\phi^2}_{g_{0}}(0)\bigr)=-U_{g}(x)\,.
\end{equation}
Since the thermal function corresponding to the Wick square is up to a constant the temperature squared \cite{BuchOjimRoos02}, we find that the condition of LTE entails an evolution equation (of Klein Gordon type) for the temperature of the system in $\mathcal{O}$\,. This is one of the instances where in a LTE state the microscopic dynamics of the system yields an equation that governs the spacetime behaviour of a macroscopic observable. It is certainly worthwhile to investigate these matters further.
\end{remark}

The next problem we wish to discuss concerns the proper choice of sets of thermal observables. Once again it turns out that relations between the fields play a critical role. Again, we will illustrate this in the model of the massless Klein Gordon field. In \cite{BuchOjimRoos02}, it was explained that the stress energy tensor $T$ on Minkowski spacetime cannot be a thermal observable. $T$ splits into two parts, one of which vanishes in all reference states, while the other is the thermal energy tensor $\epsilon$ mentioned above (which is a thermal observable). However, from the point of view of locally covariant quantum field theory, there is another reason why $T$ is not a good thermal observable.

In the massless case, with $\xi=\frac{1}{6}$ we have the following relations \cite{SchlVerc08}:
\begin{equation}\label{EMTensor}
 \begin{split}
\eta^{\mu\nu}T_{g,\mu\nu}(x)&=C_{g}(x)\,\mathbbm{1}\\
\text{and}\quad\eta^{\mu\nu}T_{g_{0},\mu\nu}(0)&=0\,,
\end{split}
\end{equation}
where $C_{g}(x)$ is determined by the local geometry of $M$ around $x$ (depending on the renormalization prescription used to determine $T_{g,\mu\nu}(x)$\,: ``conformal anomaly''). Had we chosen $T$ to be a thermal observable, and were $\omega\in\mathcal{S}(M,g)$ to be in LTE with respect to this thermal observable in the sense of Definition \ref{LTEcurved}, then:
\begin{equation}
 \begin{split}
  C_{g}(x)&=\eta^{\mu\nu}\omega(T_{g,\mu\nu}(x))\\
&=\eta^{\mu\nu}\omega_{\rho}(T_{g_{0},\mu\nu}(0))\\
&=0\,,
 \end{split}
\end{equation}
which is a contradiction - meaning that there cannot be any states in LTE with respect to $T$ (even in our sense of LTE).

The reason is that the set of observables consisting of $\mathbbm{1}$ and $T$ is not linearly independent on all spacetimes \footnote{It is not linearly independent on $(M,g)$ if $C_{g}(x)\neq 0$, cf.\ \eqref{EMTensor}.}. We therefore propose the following
\begin{condition}
 The set of thermal observables $\mathbf{T}$ must be such that $\mathbf{T}_{g}(x)$ is linearly independent in $\mathrm{span}_{\mathbb{R}}\{\mathbf{T}_{g}(x)\}$ for all $x\in M$ for all $(M,g)$\,.
\end{condition}
The sets of thermal observables $\mathbf{T}^{(2)}$ (cf.\ \eqref{SetTO}) satisfy this condition, as the field $\eta^{\mu\nu}\nabla_{\mu}\nabla_{\nu}\wick{\phi^2}$ is not a thermal observable.

To summarize: in a first step, we found that linear combinations of thermal observables \emph{need not} be thermal observables again, which led us to adopt Definition \ref{LTEcurved} as the correct characterization of LTE in curved spacetime. In a second step, we argued that linear combinations of thermal observables \emph{cannot} in general be thermal observables again. The latter point seems to be of relevance in choosing ``good'' sets of thermal observables in any model.

\quad

We end this section by briefly describing the situation in more general terms. We consider sets of thermal observables $\mathbf{T}$ such that on each spacetime $(M,g)$ at some $x\in M$ in some frame $e_{g}$\,,
\begin{equation}
  \mathbf{T}_{g}(x):=\bigl\{\tau^{(i)}_{g}(x)\bigr\}^{n}_{i=0}\quad(\text{with}~\tau^{(0)}_{g}=\mathbbm{1})\,,
\end{equation}
is linearly independent. The space $\mathrm{span}_{\mathbb{R}}\{\mathbf{T}_{g}(x)\}$ may contain further locally covariant quantum fields $\phi^{(i)}_{g}(x)$ due to relations of the types mentioned before, as shown by $\eta^{\mu\nu}\nabla_{\mu}\nabla_{\nu}\wick{\phi^2}_{g}(x)$ in our example. We define the subset $\mathbf{G}_{g}(x)$ of $\mathrm{span}_{\mathbb{R}}\{\mathbf{T}_{g}(x)\}$ to consist of the thermal observables and all additional \emph{locally covariant} quantum fields $\phi^{(i)}_{g}(x)$ which arise as their linear combinations.

We found that the definition of LTE using vector spaces of thermal observables is too restrictive in the presence of curvature, since there is an identification of $\mathbf{G}_{g}(x)$ with $\mathbf{G}_{g_{0}}(0)$ by local covariance. We are, in general, not able to extend this identification to a linear map between $\mathrm{span}_{\mathbb{R}}\{\mathbf{T}_{g}(x)\}$ and $\mathrm{span}_{\mathbb{R}}\{\mathbf{T}_{g_{0}}(0)\}$\,, due to the spacetime dependent relations among the fields.
\begin{remark}
We mention as an aside that the use of vector spaces of thermal observables is valid and equivalent to our definition in the particular case of Minkowski spacetime. In fact this is how we defined LTE in section \ref{Mink}.
\end{remark}

\section{Local Equilibrium States in Curved Spacetime}\label{Curved2}

In the following, we present a result on the existence of pointwise LTE states on curved spacetime. It naturally also applies to Minkowski spacetime as a special case. However, the methods used here are unrelated to the ones employed in Section \ref{Mink} and the discussion relies on entirely different features of QFT. We restrict attention to the simple case when our sets of fields consist of scalar fields only. However, the whole discussion as well as the results that follow also apply to the more general case.

Fix $\mathbf{T}$ and $(M,g)$\,. At some $x\in M$\,, we consider the set of thermal observables $\mathbf{T}_{g}(x)=\bigl\{\tau^{(i)}_{g}(x)\bigr\}^{n}_{i=0}$ with $\tau^{(0)}_{g}(x)=\mathbbm{1}$ and define the following convex subset of $\mathbb{R}^{n}$\,:
\begin{equation*}
 \mathfrak{T}_{g}(x):=\bigl\{\,\{\omega\bigl(\tau^{(i)}_{g}(x)\bigr)\}^{n}_{i=1}~|~\omega\in\mathcal{S}(M,g)\,\bigr\}\,.
\end{equation*}
Similarly, we define another convex subset of $\mathbb{R}^{n}$\,,
\begin{equation*}
 \mathfrak{T}_{\mathcal{C}}:=\bigl\{\,\{\omega_{\rho}\bigl(\tau^{(i)}_{g_{0}}(0)\bigr)\}^{n}_{i=1}~|~\omega_{\rho}\in\mathcal{C}\,\bigr\}\,.
\end{equation*}
Clearly $\mathfrak{T}_{g}(x)\cap\mathfrak{T}_{\mathcal{C}}\neq\emptyset$ is necessary and sufficient for the existence of $\mathbf{T}_{g}(x)$-thermal states. Let us briefly comment on the structure of the convex sets $\mathfrak{T}_{g}(x)$ and $\mathfrak{T}_{\mathcal{C}}$\,. The convex dimension\footnote{The dimension of a convex set is defined as the dimension of its affine hull. Recall that the affine hull of a subset $E$ of $\mathbb{R}^{n}$ is defined as the intersection of all affine subspaces of $\mathbb{R}^{n}$ containing $E$\,. See e.g.\ \cite{Hoer94} for details.} of $\mathfrak{T}_{g}(x)$ equals $n$ because the thermal observables are linearly independent. Therefore, $\mathfrak{T}_{g}(x)$ has non-empty interior in $\mathbb{R}^{n}$\,. It is more complicated to determine the dimension of $\mathfrak{T}_{\mathcal{C}}$\,, but it is certainly less than $n$ in case there are equations of state, i.e.\ relations among the thermal observables that show up in all reference states.

We will show in Proposition \ref{UnboundAss} that it is physically meaningful to assume that $\mathfrak{T}_{g}(x)=\mathbb{R}^{n}$, i.e.\  that for every $n$-tuple of real numbers there are states in $\mathcal{S}(M,g)$ whose expectation values in the members of $\mathbf{T}_{g}(x)$ yield exactly this tuple. Clearly we have
\begin{observation}\label{main}
 In case $\mathfrak{T}_{g}(x)=\mathbb{R}^{n}$, there are states in $\mathcal{S}(M,g)$ that are $\mathbf{T}_{g}(x)$-thermal.
\end{observation}
We now wish to relate the assumption $\mathfrak{T}_{g}(x)=\mathbb{R}^{n}$ to properties of QFT that can be checked in examples. For this, we define the \emph{numerical range} of any element $\psi\in\mathrm{span}_{\mathbb{R}}\{\mathbf{T}_{g}(x)\}$ to be the following subset of $\mathbb{R}$\,: $\{\,\omega\bigl(\psi\bigr)\,|\,\omega\in\mathcal{S}(M,g)\,\}$\,.
\begin{proposition}\label{UnboundAss}
The following two statements are equivalent.
\begin{enumerate}
 \item\label{one} The numerical range of any element in $\mathrm{span}_{\mathbb{R}}\{\mathbf{T}_{g}(x)\}$ which is not a multiple of the identity equals all of $\mathbb{R}$.
\item $\mathfrak{T}_{g}(x)=\mathbb{R}^{n}$.
\end{enumerate}
\end{proposition}

\begin{proof}
For notational reasons, we formulate the argument as if $\mathbf{T}_{g}(x)$ did not contain the unit $\mathbbm{1}$\,. However, adding $\mathbbm{1}$ is completely harmless and doesn't change our conclusions.

The following argument appears in the proof of \cite[Theorem 5]{HollWald08}. We will repeat it here for the convenience of the reader, filling in some details.

$\text{Item \ref{one}} \Rightarrow \mathfrak{T}_{g}(x)=\mathbb{R}^{n}$\,: We argue by contradiction. Assume that $\mathfrak{T}_{g}(x)$ is a proper subset of $\mathbb{R}^{n}$\,. As was noted before, $\mathfrak{T}_{g}(x)$ has non-empty interior in $\mathbb{R}^{n}$\,: $\mathfrak{T}_{g}(x)^{\circ}\neq\emptyset$\,.

We distinguish two cases. The first is that even the closure of $\mathfrak{T}_{g}(x)$ is a proper subset of $\mathbb{R}^n$. In this case we can find a support hyperplane (an affine space of dimension $n-1$ that touches the boundary of $\mathfrak{T}_{g}(x)$) which may be written as $\{v\in\mathbb{R}^n~|~L(v)=\alpha\}$ for some linear functional $L$ on $\mathbb{R}^n$ and some $\alpha\in\mathbb{R}$\,. If the components of $L$ in the standard basis of $\mathbb{R}^n$ are given by $(l_{1},\dots,l_{n})$\,, we find:
\begin{displaymath}
 l_{1}\,\omega(\tau_{g}^{(1)}(x))+\dots+l_{n}\,\omega(\tau_{g}^{(n)}(x))\leq\alpha
\end{displaymath}
for all $\omega\in\mathcal{S}(M,g)$\,. Hence, the numerical range of $l_{1}\,\tau_{g}^{(1)}(x)+\dots+l_{n}\,\tau_{g}^{(n)}(x)$ cannot be all of $\mathbb{R}$ - a contradiction.

The other case is when the closure of $\mathfrak{T}_{g}(x)$ is all of $\mathbb{R}^n$, i.e.\  when $\mathfrak{T}_{g}(x)$ is dense in $\mathbb{R}^n$. We show that in this case $\mathfrak{T}_{g}(x)$ is already all of $\mathbb{R}^n$\,. Assume the contrary and take an element $a\in\mathbb{R}^n$ such that $a\notin\mathfrak{T}_{g}(x)$. As $\mathfrak{T}_{g}(x)$ has non-empty interior, we may choose an open subset $O\subset\mathfrak{T}_{g}(x)$\,. Then $a-(O-a)$ is also open and contains some $d\in\mathfrak{T}_{g}(x)$\,, because $\mathfrak{T}_{g}(x)$ is dense. But then $d=a-(b-a)$ for some $b\in O$ and hence $a=\frac{1}{2}(d+b)$\,. Since $b,d\in\mathfrak{T}_{g}(x)$ by construction, we conclude that $a\in\mathfrak{T}_{g}(x)$ due to convexity - a contradiction.

The other direction, $\mathfrak{T}_{g}(x)=\mathbb{R}^{n}\Rightarrow\text{Item \ref{one}}$\,, is trivial and therefore omitted.
\end{proof}

To summarize: we found that in case our thermal observables span a space of "unbounded" linear forms, there exist states which are in LTE with regard to these observables. In the appendix, we will sketch how this observation can be used to show existence of LTE states in the example of the Klein Gordon field discussed previously.

\section{Conclusion}

As we have shown, the existence of LTE states can be discussed without reference to any particular model. However, we wish to point out some open questions.

In our analysis of LTE on Minkowski spacetime, we rely on the fact that the spaces of thermal observables under consideration are finite-dimensional. It would be desirable to have an argument which does not need this assumption, since one generally needs an infinite set of local thermal observables in order to attach interesting thermal functions like an entropy density to LTE states. Moreover, the regions of thermality we consider are compact and it is interesting to clarify whether this is the best one can do in this general framework.

We also pointed out some of the conceptual problems in defining LTE on curved spacetime (mainly concerning the definition of thermal observables) and gave a definition that is suited to cope with these problems. Moreover, we presented an existence result concerning \emph{pointwise} LTE states on curved spacetime. An extension of our argument to show existence of states which are in LTE in open regions of curved spacetime would be of great interest. The methods we presented may provide a first step in a proof of existence of such states.

\section*{Acknowledgments}

I am indebted to Prof.\ D.\ Buchholz for suggesting this problem and his help. Moreover, I would like to thank Dr.\ K.\ Sanders for helpful remarks and interesting conversations regarding QFT on curved spacetime. Financial support by the Graduiertenkolleg $1493$ ``Mathematische Strukturen in der modernen Quantenphysik'' is gratefully acknowledged.

\begin{appendix}
\section*{Appendix}\label{ScalingKGE}
\setcounter{section}{1}

In this appendix we wish to continue the study of $\mathbf{T}^{(2)}_{g}(x)$-thermal states for the (massless) Klein Gordon field, as introduced in section \ref{Curved}. In particular we want to sketch how Observation \ref{main} can be used to establish existence of such states.

By Proposition \ref{UnboundAss} we only need to check whether $\mathrm{span}_{\mathbb{R}}\{\mathbf{T}^{(2)}_{g}(x)\}$ consists of unbounded fields (apart from those that are multiples of $\mathbbm{1}$). We call a pointlike field \emph{unbounded} if its numerical range equals all of $\mathbb{R}$\,. Unboundedness of pointlike localized observables is something that makes sense physically: it ensures that these objects cannot be realized as proper self-adjoint operators, which is expected due to the uncertainty relations. Fewster \cite{Fews05} shows that a pointlike field is indeed unbounded if there exists a suitable scaling limit in the sense of \cite{FredHaag87} for the corresponding quantum field. We will now briefly recall this result.

Let $x\in M$ and let $(U,\kappa)$ be a chart around $x$\,. One may define a semi-group $\{\sigma_{\lambda}\}_{\lambda\in(0,1)}$ (depending on $\kappa$) of local diffeomorphisms that contract $U$ to the point $x$ if $\lambda\rightarrow 0^{+}$\,. The action of the $\sigma_{\lambda}$ is then extended to test functions on the $n$-th Cartesian power of $U$ via push-forward: $\sigma_{\lambda *}f:=f\circ\sigma^{-1}_{\lambda}$\,, $f\in C^{\infty}_{0}(U^{\times n})$\,, $\sigma_{\lambda *}f\equiv0$ outside of $U$.

Considering any scalar quantum field with corresponding hierarchy of $n$-point functions $\omega^{(n)}$ with regard to a state $\omega\in\mathcal{S}(M,g)$\,, $\omega$ is said to possess a \emph{scaling limit at $x\in M$} for the quantum field if there exists a (monotone) scaling function $N:(0,1]\rightarrow [0,\infty)$ such that the limit
\begin{equation}
 \hat{\omega}^{(n)}:=\lim_{\lambda\rightarrow 0^{+}}N(\lambda)^{n}\,\omega^{(n)}\circ\sigma_{\lambda *}
\end{equation}
exists as a distribution for all $n\in\mathbb{N}$ and is non-vanishing for some $n$. In \cite{SahlVerc01} scaling limits are generalized to vector valued quantum fields. The scaling function $N$ satisfies
\begin{equation}\label{NLambda}
 \lim_{\lambda'\rightarrow 0^{+}}\frac{N(\lambda\lambda')}{N(\lambda')}=\lambda^{\alpha}
\end{equation}
for some $\alpha\in\mathbb{R}$\,. The number $d:=4+\alpha$ is called \emph{dimension} of the field at $x$\,.

\begin{proposition}[Fewster \cite{Fews05}]\label{FewsterProp}
Let $\omega\in\mathcal{S}(M,g)$ possess a scaling limit for a quantum field $\phi$ at $x\in M$. Moreover, assume that the dimension of $\phi$ is positive and that the scaling limit two-point function $\hat{\omega}^{(2)}$ is non-trivial. Then the pointlike localized field $\phi(x)$ is unbounded.
\end{proposition}

\begin{proof}
 In \cite{Fews05} it is shown that under the stated circumstances, $\phi(x)$ is unbounded ``from below''. However, a slight modification of the argument can be used to infer that unboundedness from above is also implied. 
\end{proof}

In the case of the Wick powers of the scalar fields on some spacetime $(M,g)$\,, we have:

\begin{lemma}\label{Unbounded}
 For all $n\in\mathbb{N}$, the scaling limit of the quasi-free Hadamard state $\omega$ at $x\in M$ exists for each of the the quantum fields $\wick{(\nabla^{o(1)}\phi)(\nabla^{o(2)}\phi)\dots (\nabla^{o(n)}\phi)}$\,. The scaling limit two-point function is non-trivial and the dimension of the field is positive.
\end{lemma}
\begin{proof}
 One starts by computing the $n$-point functions of the fields
\begin{equation*}
\wick{(\nabla^{o(1)}\phi)(\nabla^{o(2)}\phi)\dots (\nabla^{o(n)}\phi)}
\end{equation*}
with respect to $\omega$\,, making use of the multiplication law of these fields (which mimics ``Wick's theorem'' \cite[eqns.\ (7) and (8)]{HollWald01}). Then, one only needs to know the scaling limit of the Hadamard parametrix, which follows from \cite[Lemma A.3]{SahlVerc01}. One finds, for example, that the scaling function for the Wick powers $\wick{\phi^{n}}(x)$ is given by $N(\lambda)=\lambda^{n-4}$, i.e.\ the dimension $d$ equals $n$. Derivatives only increase the dimension. We omit the details.
\end{proof}

Hence, all elements of $\mathbf{T}^{(2)}_{g}(x)$ except multiples of $\mathbbm{1}$ are unbounded. However, in order conclude that $\mathbf{T}^{(2)}_{g}(x)$-thermal states exist, one needs to make sure that all elements in $\mathrm{span}_{\mathbb{R}}\{\mathbf{T}^{(2)}_{g}(x)\}$ (except multiples of $\mathbbm{1}$) are unbounded. Therefore, one also has to investigate the scaling limit of linear combinations of elements of $\mathbf{T}^{(2)}_{g}(x)$\,. The only interesting question in this regard is whether the scaling limit two-point function of such a sum of fields is always non-trivial. This can fail only if there are cancellations in the scaling limit appearing among several fields that share the highest dimension in the given sum. In our case this concerns in particular the components of the thermal energy tensor. If one found that no such cancellations appear - which seems plausible -, one would be able to conclude that the required ``unboundedness properties'' are satisfied. This, by Proposition \ref{UnboundAss} and Observation \ref{main}, implies that there are states in $\mathcal{S}(M,g)$ which are $\mathbf{T}^{(2)}_{g}(x)$-thermal. These states are Hadamard by the definition of $\mathcal{S}(M,g)$\,.
\end{appendix}

\bibliography{LTE_Bibliography}

\end{document}